\documentclass[12pt]{article}

\usepackage{verbatim}
\usepackage{amsmath}
\usepackage{hyperref}
\usepackage{amssymb}
\usepackage{stmaryrd}
\usepackage{tensor}
\numberwithin{equation}{section}
\usepackage{graphicx}
\usepackage{amsthm}
\usepackage{braket}
\usepackage{accents}
\usepackage{bm}
\usepackage{tikz-cd}
\newcommand{\e}{\text{e}}
\newcommand{\bi}{\textbf{i}}
\newcommand{\ci}{i}
\newcommand{\bj}{\textbf{j}}
\newcommand{\bk}{\textbf{k}}

\newcommand{\di}{\text{d}}

\newcommand{\mbb}{\mathbb}
\newcommand{\conj}{\text{ad}}

\newtheorem{mydef}{Definition}
\numberwithin{mydef}{section}
\newtheorem{mythm}{Theorem}
\numberwithin{mythm}{section}
\newtheorem{mylemma}{Lemma}
\numberwithin{mylemma}{section}

\numberwithin{mycor}{section}
\newtheorem{myprop}{Proposition}
\numberwithin{myprop}{section}
\usepackage{authblk}

\newenvironment{sciabstract}{%
\begin{quote} \bf}
{\end{quote}}

\title{Frustrated conformal transformations and holomorphic maps on ambitwistor space}

\author{Edward B. Baker III\thanks{edwardbaker86@gmail.com}}
\affil{\textit{\small{Institute for Physical Science and Technology, University of Maryland, College Park}}}

\date{\today}

\begin{document}

\maketitle

\begin{sciabstract}
We study some aspects of conformal transformations in the context of twistor theory, leading to the definition of a \textit{frustrated conformal transformation}.  This equation relies on two instantons for the left and right copies of $Sp(1)$, one being self-dual and the other anti-self-dual.  Solutions to this equation naturally generate maps on ambitwistor space due to an analog of the Penrose-Ward correspondence.  A solution based on the BPST instanton is presented.
\end{sciabstract}

\section{Introduction}

Due to the success of the theory of complex analysis and complex manifolds in two dimensions, it is natural to search for an analog of holomorphicity in higher dimenions, which may be different than holomorphicity of multiple complex variables.  If one uses equivalence with conformal transformations to try to achieve this goal, there is an obstacle due to Liouville's theorem, which states that conformal transformations are highly restrictive in dimension greater than two.  One way of attempting to bypass this restriction is to search for generalizations of different definitions of holomorphicity to higher dimensions, which may not necessarily be conformal maps.  An early attempt to achieve this goal was made by Fueter \cite{Fueter}\cite{Sudbery}, related to the original work by Hamilton \cite{hamilton1866elements}, which introduced a notion of quaternion holomorphicity.   The Cauchy-Riemann-Fueter equation bears a clear resemblance to the Cauchy-Riemann equations, but is lacking in a number of respects.  One important property that is lacking is that these functions are not closed under composition or multiplication, and therefore do not form a group of transformations.  For this reason and others, the theory of quaternion regular functions is not an adequate generalization of holomorphicity to four dimensions.

Despite its shortcomings, this approach has been formalized and extended in a number of directions.  Dominic Widdows studied quaternion valued forms on hypercomplex manifolds \cite{alggeom}, generalizing the tools of algebraic geometry to this context, which was related to previous work by Dominic Joyce \cite{Joyce98hypercomplexalgebraic}.  Quaternionic coordinates and maps were also introduced for a restrictive definition of quaternionic manifolds \cite{kulkarni1978}, which was later generalized to a class of quaternionic manifolds without reference to a specific coordinate system \cite{Salamon1986}.  These results are related to the study of manifolds with special holonomy, which has applications to string theory and topology \cite{ellingsrud2002calabi}. 
 
 A powerful way of understanding the framework of four dimensional geometry is through twistor theory \cite{Penrose197765}\cite{ward1991twistor}\cite{dunajski2010solitons}, and the related study of self-duality \cite{atiyah1979geometry}\cite{10.2307/79638}.  In particular, the twistor space of a four dimensional quaternionic manifold is an integrable complex manifold, and conformal maps on this manifold induce holomorphic maps on its twistor space.  Quaternionic manifolds in four dimensions have a self-dual conformal structure, and so the integrability of the twistor space is related to the self-duality of the underlying base space.  These and other relationships between conformal transformations, self-duality, and integrable complex structures provide powerful ways of understanding how holomorphicity manifests itself in higher dimensions.  

In this paper we analyze some aspects of these ideas, introducing some new tools and perspectives along the way.  This analysis leads to the definition of a \textit{frustrated conformal transformation}, whose construction relies on a pair of self-dual and anti-self-dual $Sp(1)$ connections, which are in some sense dual to each other.  Through a construction analogous to the Penrose-Ward correspondence, these functions generate holomorphic maps of ambitwistor space, which  were introduced in the study of general solutions to the Yang-Mills equations \cite{WITTEN1978394}\cite{ISENBERG1978462}.  These maps are naturally closed under composition, and provide a different possible notion of generalized holomorphic functions in four dimensions.   A special case based on the BPST instanton \cite{BELAVIN197585} is discussed.

\section{Quaternionic embeddings}
\subsection{Quaternion projective space}
We begin our discussion by analyzing the classic example of twistor space over the quaternion projective line $\mathbb{H}\mathbb{P}^1$. $\mathbb{H}\mathbb{P}^1$ can be defined by homogeneous coordinates $(q_1,q_2)$ defined up to left multiplication by $u\in\mbb{H}^*$.\footnote{We use equivalence under left multiplication for notational reasons, but the opposite convention is related by a change in orientation.} This space can be understood in terms of complex numbers through the fibration $\pi_s:\mathbb{C}\mathbb{P}^3\rightarrow \mathbb{H}\mathbb{P}^1$, defined by 
\begin{equation}\label{eq:fibration}
\pi_s:(z_1,z_2,z_3,z_4)\rightarrow(z_1+z_2 \bj ,z_3+z_4 \bj),
\end{equation}
where $(z_1,z_2,z_3,z_4)$ are homogeneous coordinates on $\mathbb{C}\mbb{P}^3$, and the imaginary unit of the variables $z_i$ is identified with the quaternion $\bi$.  This is the projection map of a fibre bundle where the fibres are copies of $S^2$, and was used extensively in the characterization of instantons on the four sphere, which is the prototype for the ADHM construction \cite{atiyah1979geometry}\cite{ATIYAH1978185}.  In this case, $\mbb{C}\mbb{P}^3$ can be viewed as the twistor space of $\mbb{H}\mbb{P}^1$.  

It will prove useful to locally view the twistor space $\mbb{C}\mbb{P}^3$ as a subset of the product space $\mbb{H}\otimes \mbb{I}_u$, where $\mbb{I}_u=\{q\in\mbb{H}:\lvert q\rvert =1, \text{Re}(q)=0\}$ denotes the unit imaginary quaternions.  To see this, consider the coordinate chart in $\mbb{C}\mbb{P}^3$ defined by $z_4=1$, and the chart in $\mbb{H}\mbb{P}^1$ defined by $q_2=1$.  In these charts, the differential of the map \eqref{eq:fibration} becomes 
\begin{equation*}
\di \pi_s(z_1,z_2,z_3)=-(z_3+\bj)^{-1}\di z_3 (z_3+\bj)^{-1}(z_1+z_2 \bj)+(z_3+\bj)^{-1}(\di z_1+\di z_2 \bj),
\end{equation*}  
where the tangent space $T_p\mbb{H}$ is identified with $\mbb{H}$. The complex structure on $\mbb{C}\mbb{P}^3$ acts on the differentials $\di z_i$ through multiplication by $\bi$, which is equivalent to left multiplication of the differential by $\eta(z_3)$, where $\eta:\mbb{C}\rightarrow \mbb{I}_u$ is defined by 
\begin{equation}\label{eq:stereographic}
\eta(z)=(z+\bj)^{-1}\bi(z+\bj),
\end{equation}
and can be identified with a complex structure on $\mbb{H}$. Identifying $\mbb{I}_u$ with the Riemann sphere, the function $\eta(z)$ gives a stereographic projection, which can be readily verified.  We are therefore motivated to define the local trivialization $\phi:U\rightarrow V$ for $U\in \mbb{C}\mbb{P}^3$ and $V\in\mbb{H}\otimes\mbb{I}_u$ given by
\begin{equation}\label{eq:phitriv}
\phi:(z_1,z_2,z_3)\rightarrow\Big((z_3+\bj)^{-1}(z_1+z_2\bj),\eta(z_3)\Big),
\end{equation}
in the chosen coordinate patches.  If one instead chooses coordinates such that $q_1=1$, then the  transition function takes the form
\begin{equation}\label{eq:qptransition}
\tau:(q,\eta)\rightarrow (q^{-1},q^{-1}\eta q),
\end{equation}
which is analogous to the transition function $z\rightarrow z^{-1}$ of the Riemann sphere, with an additional rotation on the fibre. Calculating the differential of $\eta(z_3)$,
\begin{equation*}
\di \eta(z_3) = (z_3+\bj)^{-1}[\bi,\di z_3 (z_3+\bj)^{-1}](z_3+\bj),
\end{equation*} 
it is clear that the induced complex structure on the fibre is also given by left multiplication by $\eta(z_3)$.

This example motivates us to consider the space $\mbb{H}\otimes \mbb{I}_u$ as a natural coordinate system for the twistor space of a four-dimensional quaternionic manifold.  As is the case for any quaternionic manifold, the twistor space has an induced complex structure, which is embodied in this coordinate system through left multiplication of the tangent space by $\eta\in\mbb{I}_u$.  In the case considered above, the transition functions for these coordinates can be visualized with the diagram
\begin{equation} \label{eq:comm1}
\begin{tikzcd}
U_1\in\mbb{C}\mbb{P}^3 \arrow[r, "\hat{\tau}"] \arrow[d, "\phi"']
& U_2\in\mbb{C}\mbb{P}^3 \arrow[d, "\phi"] \\
V_1\in\mathbb{H}\otimes\mathbb{I}_u \arrow[r, "\tau"]
& V_2\in\mbb{H}\otimes\mbb{I}_u
\end{tikzcd}
\end{equation}
where $\phi$ is defined by equation \eqref{eq:phitriv}, and $\hat{\tau}$ is holomorphic.  The demand that $\tau$ is holomorphic with respect to the defined complex structure is equivalent to demanding that the diagram commutes, which follows from the fact that the other maps are also holomorphic. This picture will be useful in the analysis that follows.

\subsection{M{\"o}bius transformations and holomorphic embeddings}
To gain some more familiarity with quaternions, we first consider how M{\"o}bius transformations can be represented with the identification of $\mbb{C}_\infty$ with $\mbb{I}_u$.    In homogeneous coordinates, M{\"o}bius transformations take the form $\mathcal{M}:(z_1,z_2)\rightarrow (az_1+bz_2,cz_1+dz_2)$ with $ad\neq bc$. Define the map $\eta:\mbb{C}\mbb{P}^1\rightarrow \mbb{I}_u$ by
\begin{equation}\label{eq:stereographic2}
\eta(z_1,z_2)=(z_1+z_2\bj)^{-1}\bi(z_1+z_2\bj),  
\end{equation}
which reduces to equation \eqref{eq:stereographic} in the coordinate patch $z_2=1$.
Applying this map to both sides of $\mathcal{M}$, we find that the M{\"o}bius transformation takes the form 
\begin{equation}\label{eq:etamob}
\mathcal{M}_{\alpha,\beta}:\eta\rightarrow (\alpha+\eta\beta)^{-1}\eta(\alpha+\eta\beta),
\end{equation}
where $2\alpha=a+c\bj+(\bj b+d)^*$ and $2\bi\beta=a+c\bj-(\bj b+d)^*$.  
In this form, the non-singular condition becomes $\beta\alpha^{-1}\not\in\mbb{I}_u$. Furthermore, the invariance of $a, b, c, d$ under simultaneous multiplication by $u\in\mbb{C}^*$ translates into a scaling and rotational symmetry between $\alpha$ and $\beta$, and the rotation subgroup consists of transformations with $\beta=0$. 

In fact, this representation of the M{\"o}bius group is realized as a subgroup of the ``biquaternions'', denoted $\mbb{H}_\mbb{C}$, where $\alpha$ and $\beta$ are identified with the real and imaginary parts of a complexified quaternion. The restricted Lorentz group $SO^+(1,3)$, which is isomorphic to the M{\"o}bius group $PSL(2, \mbb{C})$, can be embedded in the biquaternions, explaining this relationship \cite{0143-0807-5-1-007}.   In this context, the demand $\beta\alpha^{-1}\not\in\mbb{I}_u$ is equivalent to requiring that $\alpha+\ci\beta$ is not a zero divisor.  In the biquaternion representation, $SO^+(1, 3)$ is realized as a pair $\alpha, \beta \in \mbb{H}$ satisfying
\begin{equation*}
(\alpha+\ci\beta)(\bar{\alpha}+\ci\bar{\beta})=1,
\end{equation*}
where here $\ci$ is the unit of complexification for the biquaternions. The action of this group on a
``Minkowskian quaternion'' $q_m$ of the form $q_m=x_0+\ci(\bi x_1+\bj x_2 + \bk x_3)$ is given by
\begin{equation*}
q_m\rightarrow (\alpha+\ci\beta)q_m(\bar{\alpha}-\ci\bar{\beta}),
\end{equation*}
which is analogous to equation \eqref{eq:etamob}.   
 
 In this context, the Lie algebra of $SO^+(1, 3)$ corresponds to quaternions of the form $\gamma+\ci\delta$ where $\gamma$ and $\delta$ are pure imaginary. With an eye towards analyzing the transition map $\tau$ from diagram \eqref{eq:comm1}, we can consider a continuously varying M{\"o}bius transformation on the space $\mbb{H}$, for which the differential of the M{\"o}bius transformation can be written
\begin{equation}\label{eq:moblie}
(\bar{\alpha}+\ci \bar{\beta})(\di \alpha + \ci \di \beta)=\gamma+\ci\delta.
\end{equation}
This defines $\gamma,\delta\in\Gamma(\mbb{I}\otimes \Lambda^1)$, where $\Lambda^1$ is the space of one forms and $\mbb{I}$ are the imaginary quaternions.
Working out the action of the generators on $q_m$, we find that $\gamma$ is the generator of rotations and $\ci\delta$ is the generator of boosts. The relationships between these ideas will be developed further in the following sections.

It will also be relevant to consider embeddings of the Riemann sphere into the space $\mbb{C}\mbb{P}^3$. In the context of the diagram \eqref{eq:comm1}, these maps will correspond to transition functions $\tau:\mbb{H}\otimes\mbb{I}_u\rightarrow\mbb{H}\otimes\mbb{I}_u$ at fixed $q$. To this end, consider a degree one embedding $\iota:\mbb{C}\mbb{P}^1\hookrightarrow\mbb{C}\mbb{P}^3$.   In homogeneous coordinates, these linear transformations can be written similarly to M{\"o}bius transformations
\begin{equation}\label{eq:cp1tocp3}
(z_1,z_2)\rightarrow (az_1+bz_2,cz_1+dz_2,\tilde{a}z_1+\tilde{b}z_2,\tilde{c}z_1+\tilde{d}z_2).
\end{equation}   
Applying equations \eqref{eq:stereographic2} to the domain and \eqref{eq:phitriv} to the range of this embedding, this can be rewritten 
\begin{equation} \label{eq:tline}
\tau_q(\eta)=\Big(h+\mathcal{M}_{\alpha,\beta}(\eta)\rho,\mathcal{M}_{\alpha,\beta}(\eta)\Big).
\end{equation}
The details of this calculation are presented in appendix \ref{sec:emriemann}, along with explicit formulas for $h,\rho\in\mbb{H}$ as a function of the embedding parameters.  Through this calculation, we see that a line in the context of algebraic geometry more closely resembles a sphere in this coordinate system, where $h$ is the center of the sphere and $\rho$ is a generalized quaternionic radius.  

For the following analysis, we define 
\begin{mydef}
Given a domain $U\in\mbb{H}$, a map $\tau:U\otimes \mbb{I}_u\rightarrow \mbb{H}\otimes \mbb{I}_u$ is twistor conformal if at fixed $q\in U$, the map restricts to a degree one holomorphic embedding of the form \eqref{eq:tline}.
\end{mydef}
These maps will turn out to be somewhat trivial, given by generalized M{\"o}bius transformations.  They will, however, lead to a nontrivial generalization, and therefore merit a close investigation.

\section{Twistor conformal maps}
\subsection{Analysis on quaternions}\label{sec:qanalysis}
We will now consider the transition functions $\tau$ from diagram \eqref{eq:comm1} at fixed $\eta$, which are holomorphic maps $\tau_\eta:\mbb{H}\rightarrow\mbb{H}\otimes\mbb{I}_u$, or 
\begin{equation}\label{eq:taueta}
\tau_\eta:q\rightarrow \big(f_\eta(q),m_\eta(q)\big),
\end{equation} 
We will want to demand that these functions are holomorphic with respect to the complex structures defined. To understand this more clearly, we first develop some useful tools in quaternionic analysis with a preferred complex structure.

The differential of a quaternion valued function $f(q)$ can be written
\begin{equation*}
\di f = \partial_{\bar{q}} \di \bar{q} f + \di q\partial_{q}f
\end{equation*}
where we have defined the derivative and one form
\begin{equation*}
\partial_q = \frac{1}{2}(\partial_0-\bi\partial_1-\bj\partial_2-\bk\partial_3), \ \ \di q = \di x_0+\bi\di x_1+\bj \di x_2+\bk \di x_3.
\end{equation*}
This expression for the differential is analogous to the formula $\di=\di z \partial_z+\di\bar{z}\partial_{\bar{z}}$ familiar from complex analysis.
The following identities will be useful for our analysis
\begin{align}
\text{Re}(q)=\frac{1}{4}\sum_{\mu} \e_\mu q \bar{\e}_{\mu}, \ \ \ \bar{q}=-\frac{1}{2}\sum_\mu \e_\mu q \e_\mu, \label{eq:qident12}
\end{align}
where we have identified $\e_\mu=\{1,\bi,\bj,\bk\}$ for $\mu\in [0,3]$. Using these identities, the differential can be rewritten 
\begin{equation} \label{eq:quatdif}
\di f = \frac{1}{2}\sum_\mu \e_\mu\di q\partial_q \bar{\e}_\mu f.
\end{equation}
 With this differential, we have implicitly identified the tangent space of the target space with the quaternions, so the complex structure acts on the target space through left multiplication by $m_\eta(q)$.  The complex structure acts on the differential $\di q$ through
\begin{equation*}
I_\eta^r(\di q)=\eta \di q.
\end{equation*}
which acts on the cotangent space, and so according to our conventions should act on the right, which is the reason for the superscript. Using the complex structure of the target space, a function that is holomorphic on these coordinates satisfies
\begin{equation}\label{eq:fholdef}
I_\eta^r(\di f)=m_\eta \di f.
\end{equation}
This equation can be understood as requiring that $\di f$ is contained in the holomorphic
cotangent space of $\mbb{H}$ for the given complex structure. For example, in the case $\eta=m_\eta(q)=\bi$, we find $\partial_q \bj f=\partial_q \bk f=0$. 
Writing $q=z_1+z_2\bj$ and $f_{\bi}(q)=f_1(z_1,z_2)+f_2(z_1,z_2)\bj$, this reduces to the demand that $f_1(z_1,z_2)$ and $f_2(z_1,z_2)$ are holomorphic with respect to $z_1$ and $z_2$.  In the more general case there is a similar rewriting, but equation \eqref{eq:fholdef} is the most invariant way of writing the condition.

We will also be interested in properties of quaternion valued two forms.  A short calculation shows
\begin{align*}
&\di q\wedge \di \bar{q}=-2\sum_i \e_i\left(\di x_0\wedge \di x_i+\sum_{jk}\epsilon_{ijk}\di x_j\wedge\di x_k\right),\\
&\di q\wedge \e_i\di \bar{q}=2\di x_0\wedge \di x_i -2\sum_{jk}\epsilon_{ijk}\di x_j\wedge\di x_k,
\end{align*}
with similar equalities for the opposite ordering.  In this way, we see that $\di q\wedge \di \bar{q}$ is self-dual and $\di q\wedge \e_i\di \bar{q}$ is real and anti-self-dual under the hodge-star operation (with the volume form $\di x_0\wedge \di x_1\wedge \di x_2\wedge \di x_3$). Simiarly $\di \bar{q}\wedge \di q$ is anti-self dual and $\di \bar{q}\wedge \e_i\di q$ is self-dual. 
Using the above results in addition to equation \eqref{eq:qident12}, one finds the useful identity
\begin{equation}
\di q \wedge w\di q = -\frac{1}{2}\Big(\bar{w}\di \bar{q}\wedge \di q+\di q \wedge \di \bar{q} \bar{w}\Big) \ \ \ \forall w\,\in\mbb{H}_{\mbb{C}}. \label{eq:qident3}
\end{equation} 
There are also variants of these identities for the different orderings of $\di q$, $\di \bar{q}$ and $w$.

\subsection{Conformal coordinate transformations}\label{sec:cco}

For a twistor conformal map, the map $m_\eta(q) : \mbb{H} \rightarrow \mbb{I}_u$ defined in equation \eqref{eq:taueta} is a M{\"o}bius transformation. Therefore, in accordance with equation \eqref{eq:etamob}, we introduce quaternion valued functions $\alpha(q),\beta(q)$ such that $m_\eta(q) = \mathcal{M}_{\alpha,\beta}(\eta)$. We want this map to be holomorphic for all $\eta$, which implies
\begin{equation*}
I^r_\eta(\di m_\eta)=m_\eta \di m_\eta.
\end{equation*}
Using equation \eqref{eq:etamob}, we can write the differential in the form
\begin{equation}\label{eq:meq}
\di m_\eta=(\alpha+\eta\beta)^{-1}[\eta,(\di \alpha +\eta \di \beta)(\alpha+\eta\beta)^{-1}](\alpha+\eta\beta),
\end{equation}
with which the condition of holomorphicity can be rewritten
\begin{equation}\label{eq:holcond1}
I^r_\eta([\eta,\gamma+\eta\delta])=\eta[\eta,\gamma+\eta\delta],
\end{equation}
where $\gamma$ and $\delta$ are defined by equation \eqref{eq:moblie}.  In this way, the biquaternion representation of $\mathfrak{s}\mathfrak{o}(1,3)$ appears naturally in the differential of the map $\tau$ at fixed $\eta$.

To understand this condition for all $\eta$, it will be useful to introduce some notation. First, define the operator
\begin{equation*}
\mathcal{I}_\eta=\eta-I^r_\eta,
\end{equation*}
which projects onto the holomorphic cotangent space for $\omega\in\Gamma(\mathbb{H}_{\mbb{C}}\otimes \Lambda^1)$, and complex structure defined by $\eta$.  The Lie algebra generated by these operators was studied by Widdows \cite{alggeom}, leading to a definition of the ``q-holomorphic'' cotangent space.\footnote{In that context, the cotangent space $\mbb{H}\otimes \Lambda^1$ was being considered, and the conventions used here are opposite regarding the left and right $\mbb{H}$ action.} We also define the operators
\begin{equation*}
\pi_\eta^{\pm}=\ci\pm\eta ,
\end{equation*}
which project onto the $\pm\ci$ eigenspaces of multiplication by $\eta$, and satisfy the identities  
\begin{equation*}
\pi_\eta^\pm\pi_\eta^\pm=2\ci \pi_\eta^\pm, \ \ \pi_\eta^+\pi_\eta^-=0.
\end{equation*}
In the context of twistor theory, these operators are related to coordinates $\pi_A$ on twistor space, which is the reason for the similar notation.  This relationship will be developed further.

A short calculation shows that equation \eqref{eq:holcond1} can be written
\begin{equation}\label{eq:holcond1b}
\mathcal{I}_\eta \pi_\eta^+[\eta,\bar{\varphi}\di\varphi ]=0,
\end{equation}
where we have defined 
\begin{equation*}
\bar{\varphi} = \alpha+\ci \beta,
\end{equation*}
and used the definition \eqref{eq:moblie} in the form $\bar{\varphi}\di \varphi = \gamma+\ci\delta$.   Here, the quaternion conjugation is included for agreement with conventions in twistor theory.
This way of writing the equation manifests the different projections that are being applied to the differential form before setting the result to zero. We will want this to be true for all $\eta$, which will define a subspace of the biquaternion valued cotangent space.

In order to find the most general solution for all $\eta$, we first focus on the simultaneous projections of $[\eta,\cdot]$ and $\mathcal{I}_\eta$. To this end, we consider the Lie algebra consisting of elements
\begin{equation*}
\mathcal{J}_\eta(\omega)=\eta\omega-\omega\eta-I^r_\eta(\omega).
\end{equation*}
 The solutions to equation \eqref{eq:holcond1b} that are true for all $\eta$ must transform invariantly under the Lie group generated by $\mathcal{J}_\eta$. Decomposing this into irreducible components\footnote{Here we use physics notation for these representations. More background for this technique can be found in \cite{alggeom}, where the mathematics notation is used.}, we find $\frac{\mathbf{1}}{\mathbf{2}}\otimes\frac{\mathbf{1}}{\mathbf{2}}\otimes\frac{\mathbf{1}}{\mathbf{2}}=\frac{\mathbf{3}}{\mathbf{2}}\oplus 2\cdot \frac{\mathbf{1}}{\mathbf{2}}$.  
In this case the relevant subspace is given by the two spin half representations. To see this, consider $\omega\in \Gamma(\mbb{H}_{\mbb{C}}\otimes T^*\mbb{H})$ of the form
\begin{equation*}
\omega=\di q\xi_1-\bar{\xi}_2 \di \bar{q},
\end{equation*}
for $\xi_1,\xi_2\in\mathbb{H}_{\mbb{C}}$.\footnote{In general, for biquaternions we will denote $\bar{\xi}$ as the ``quaternion conjugate'' and $\xi^*$ as the conjugate with respect to $\ci$.} One can readily check that choosing $\bar{\varphi}\di\varphi=\omega$ satisfies equation \eqref{eq:holcond1b}, and additionally this subgroup transforms invariantly under the Lie group.  Note that the normalization condition $\varphi\bar{\varphi}=1$ implies $\xi_1=\xi_2$.

This formalism can also be applied to the function $f_\eta(q)$ from the map \eqref{eq:taueta}, which is equal to $f_\eta(q) = h(q) + m_{\eta}(q)\rho(q)$, where $h(q)$ and $\rho(q)$ are defined in equation  \eqref{eq:tline}. Defining 
 \begin{equation*}
\kappa = h+\ci \rho,
\end{equation*}  
and using equation \eqref{eq:fholdef}, we find 
\begin{equation}\label{eq:holcond}
\mathcal{I}_\eta \pi_\eta^+\bar{\varphi} \,\di \kappa=0.
\end{equation}
With this form, we can use the same argument to show that $\di \kappa$ transforms invariantly under the Lie algebra generated by $\mathcal{I}_\eta$. In this case, addition of angular momentum gives $\frac{\mathbf{1}}{\mathbf{2}}\otimes\frac{\mathbf{1}}{\mathbf{2}}=\mathbf{1}\oplus \mathbf{0}$,
and the relevant solution is given by the $\mathbf{0}$ representation. This implies that
\begin{equation}\label{eq:kappaeq}
\bar{\varphi}\di \kappa = \di q \nu,
\end{equation}
for $\nu\in \mbb{H}_{\mbb{C}}$.

To clean up the notation, define the M{\"o}bius transformation $\mathcal{M}_{\varphi}$ for $\varphi \in \mathbb{H}_{\mbb{C}}$ as
\begin{equation*}
\mathcal{M}_{\varphi}\equiv \mathcal{M}_{\Re(\varphi),\Im(\varphi)}.
\end{equation*}
We then summarize the above results with
\begin{mythm}\label{tcon}
A map $\tau$ is twistor conformal if and only if it takes the form 
\begin{equation}\label{eq:tconform}
\tau(q,\eta)=\big(\Re(\kappa)+\mathcal{M}_{\bar{\varphi}}(\eta)\Im(\kappa),\mathcal{M}_{\bar{\varphi}}(\eta)\big), \ \ \ q\in\mbb{H}, \eta\in\mbb{I}_u,
\end{equation}
where $\varphi\bar{\varphi}=1$ and $\kappa(q),\varphi(q)\in\mbb{H}_{\mbb{C}}$ satisfy the conditions  
\begin{align}
&\di \kappa =\varphi \di q \nu, \label{eq:tconcond1} \\
&\bar{\varphi} \di \varphi =\frac{1}{2}(\di q \xi -\bar{\xi}\di\bar{q}), \label{eq:tconcond2}
\end{align}
for $\nu(q),\xi(q)\in\mbb{H}_{\mbb{C}}$.  Here $\bar{\varphi}$ denotes the quaternionic conjugate of $\varphi$, and $\Re$, $\Im$ denote the real and imaginary parts with respect to $\ci$.
\end{mythm}

It turns out that twistor conformal maps are very restrictive, which is clear from the following theorem.

\begin{myprop}\label{tconfprop}
Twistor conformal functions are M{\"o}bius transformations on the space $\mbb{H}_{\mbb{C}}$, taking the form
\begin{equation}\label{eq:bqmob}
\kappa(q)=(\alpha q+\beta)(\gamma q+\delta)^{-1},
\end{equation}
for $q,\alpha,\beta,\gamma,\delta\in\mbb{H}_{\mbb{C}}$.
\end{myprop}
\begin{proof}  
First, taking the exterior derivative of equation \eqref{eq:tconcond1}, we find
\begin{equation*}
\bar{\varphi}\di \varphi \wedge\di q-\di q\wedge (\di \nu)\nu^{-1}=0.
\end{equation*}
From equations  \eqref{eq:qident3} and \eqref{eq:tconcond2}, this implies 
\begin{equation*}
\di \nu \, \nu^{-1}=\frac{1}{2}(\di \bar{q} \bar{\xi}+3\xi \di q)
\end{equation*}
Now, define the functions
\begin{equation*}
z^4=\nu\bar{\nu}, \ \ \chi = z \varphi, \ \ \psi = \frac{\nu}{z} . 
\end{equation*}
where the particular branch of the complex function $(\nu\bar{\nu})^{1/4}$ will not be important for this analysis, as the result will hold locally.  With this definition equation \eqref{eq:tconcond1} becomes
\begin{equation}\label{eq:tconfeqcp}
\di \kappa= \chi \di q \psi.
\end{equation}
After a short calculation, we find
\begin{equation*}
 \chi^{-1} \di \chi =\di q \,\xi , \ \ \di \psi \, \psi^{-1}=\xi \,\di q.
\end{equation*}
Quaternion valued functions whose differential can be written $\di f(q)=\di q \,g(q)$ are necessarily linear  of the form $f(q)=q\alpha+\beta$ \cite{Sudbery}, so we conclude that the two functions are given by
\begin{equation*}
\chi^{-1}=q\tilde{\alpha}+\tilde{\beta}, \ \ \psi^{-1}=\tilde{\gamma}q+\tilde{\delta}
\end{equation*}
for $\tilde{\alpha},\tilde{\beta},\tilde{\gamma},\tilde{\delta}\in\mbb{H}_{\mbb{C}}$. This implies 
\begin{equation*}
\di \kappa = (q\tilde{\alpha}+\tilde{\beta})^{-1}\di q (\tilde{\gamma}q+\tilde{\delta})^{-1},
\end{equation*}
which is the differential of \eqref{eq:bqmob} if we identify $\tilde{\alpha}=(\alpha\gamma^{-1}\delta-\beta)^{-1}$, $\tilde{\beta}=(\alpha-\beta\delta^{-1}\gamma)^{-1}$, $\tilde{\gamma}=\gamma$, $\tilde{\delta}=\delta$.
\end{proof}

\subsection{Relation to twistor theory}\label{sec:rtwist}

We will now consider extending the map  defined in theorem \ref{tcon} to the domain $q\in \mathbb{H}_{\mbb{C}}$.  This space can be identified with the complexified Minkowski space $\mbb{C}^4$ considered in twistor theory, and many of the constructions used in twistor theory have analogues in this formulation. 
A relevant description of twistor theory is given in the books of Ward and Wells \cite{ward1991twistor} and Dunajski \cite{dunajski2010solitons}. Define the correspondence space $\mathcal{F}=\mbb{H}_\mbb{C}\otimes\mbb{I}_u$, the twistor space $\mathcal{P}\mathcal{T}=\mbb{H}\otimes\mbb{I}_u$ and the double fibration 
\begin{equation}\label{eq:doublefib}
\mbb{H}_\mbb{C}\xleftarrow{r} \mathcal{F}\xrightarrow{a}\mathcal{P}\mathcal{T}.
\end{equation}
Here the map $r$ is the natural projection map, and $a$ maps onto the $\alpha$-planes of $\mathcal{F}$.  The coordinates used here are chosen such that the Euclidean slice is identified with the real surface $\Im(q)=0$, and are related to the usual four complex dimensional coordinate system $(z,\tilde{z},w,\tilde{w})$ by 
\begin{equation*}
z=\frac{q_0+\ci q_3}{\sqrt{2}}, \ \tilde{z}=\frac{q_0-\ci q_3}{\sqrt{2}}, \ w = \frac{q_1+\ci q_2}{\sqrt{2}}, \ \tilde{w}=\frac{-q_1+\ci q_2}{\sqrt{2}}
\end{equation*} 
In this coordinate system, the metric is 
\begin{equation*}
\di s^2 = \di q \odot \di \bar{q},
\end{equation*}
where $\odot$ denotes the symmetric product.  
An $\alpha$-plane is defined as a set of points $q\in\mbb{H}_\mbb{C}$ satisfying
\begin{equation*}
\pi_\eta^+ (q-c) = 0 
\end{equation*}
for some fixed $c\in\mbb{H}_{\mbb{C}}$.  This is the quaternion equivalent of the twistor definition of an $\alpha$-plane $x^{AA'}\pi_{A'}=\omega^A$, which shows how $\pi_\eta^+$ is related to the coordinate $\pi_{A'}$ defined in twistor theory.  Two points $q_1$ and $q_2$ lie on the same $\alpha$-plane if they satisfy the relation
\begin{equation}\label{eq:pointequiv}
q_2 =q_1+ \pi_\eta^- \delta,
\end{equation}
which is analogous to the translation $x^{AA'}\rightarrow x^{AA'}+\kappa^A\pi^{A'}$ from twistor theory.

Consider a map $\hat{\tau}:\mathcal{F}\rightarrow\mathcal{F}$ that descends to the space $\mbb{H}_\mbb{C}$ under the projection $r$. We will also want this map to descend to the twistor space, in the sense that the diagram 
\begin{equation} \label{eq:comm2}
\begin{tikzcd}
\mathcal{F} \arrow[r, "\hat{\tau}"] \arrow[d, "a"']
& \mathcal{F} \arrow[d, "a"] \\
\mathcal{P}\mathcal{T}\arrow[r, "\tau"]
& \mathcal{P}\mathcal{T}
\end{tikzcd}
\end{equation}
commutes, yielding a twistor conformal function $\tau$. On the correspondence space, the functions $\kappa(q),\varphi(q)$ will be the same as in theorem \ref{tcon}, but are extended to the domain of $\mathcal{F}$. As before, define the function $m_\eta(q)=\mathcal{M}_{\bar{\varphi}}(\eta)$. This function descends to the twistor space only if $d m_\eta (\pi_\eta^-\delta)=0$, where $\di m_\eta(\pi_\eta^-\delta)$  is a contraction of a one form with a vector under the identification of $T\mbb{H}_{\mbb{C}}$  with $\mbb{H}_{\mbb{C}}$. This condition can be rewritten
\begin{equation}\label{eq:tlineeq1}
\pi_\eta^+[\eta,\bar{\varphi}\di \varphi(\pi_\eta^-\delta)]=0 \ \ \forall\delta\in\mbb{H}_{\mbb{C}}.
\end{equation}
 Similarly, $\kappa(q)$ descends to the twistor space only if 
\begin{equation}\label{eq:tlineeq2}
\pi_{m_\eta}^+\di\kappa(\pi_\eta^-\delta)=0 \ \ \forall\delta\in\mbb{H}_{\mbb{C}}.
\end{equation}

These conditions fit in naturally with the conditions of holomorphicity derived in the previous section.  To see this, first define
\begin{mydef}
A one-form $\omega\in \Gamma(\mbb{H}_{\mbb{C}}\otimes \Lambda^1)$ is holomorphic if it satisfies the condition $\omega(\ci\delta)=\ci\omega(\delta)$ for all $\delta\in\mbb{H}$. A function $f\in\Gamma(\mbb{H}_\mbb{C})$ defined on $\mathcal{F}$ is said to be holomorphic if its differential is holomorphic.
\end{mydef}
We then prove 
\begin{mylemma}\label{hollemma}
For any holomorphic one-form  $\omega$, the condition
\begin{equation*} 
\pi_\eta^+\omega(\pi_\eta^-\delta)=0 \  \forall\delta\in\mbb{H}_{\mbb{C}}
\end{equation*}
is equivalent to
 \begin{equation*}
 \mathcal{I}_\eta \pi_\eta^+ \omega=0.
 \end{equation*}
Furthermore,  only holomorphic forms can satisfy the first condition for all $\eta$.
\end{mylemma}
\begin{proof}
First assume that $\omega$ is holomorphic and satisfies $\mathcal{I}_\eta \pi^+_\eta \omega(\delta) = 0$. By linearity, we have $\pi^+_\eta \omega(\pi^-_\eta \delta) = \ci \pi_\eta^+ \omega(\delta)-\pi^+_\eta \omega(\eta\delta)$. By definition $\omega(\eta\delta) = I_\eta^r \omega(\delta)$, and also $\ci \pi_\eta^+ = \eta \pi^+_\eta$. Therefore, $\pi_\eta^+ \omega(\pi^-_\eta \delta) = \mathcal{I}_\eta \pi^+_\eta \omega(\delta) = 0$. Now assume that $\pi^+_\eta \omega(\pi^-_\eta \delta) = 0$, the same calculation in reverse shows $\mathcal{I}_\eta \pi^+_\eta \omega(\delta) = 0 \ \forall \delta \in \mbb{H}_{\mbb{C}}$, so $\mathcal{I}_\eta \pi^+_\eta \omega= 0$.

Now consider a general one-form $\omega$, and decompose it into holomorphic and anti-holomorphic components $\omega = \omega^+ + \omega^-$. If $\omega$ satisfies the first condition for all $\eta$, then so do $\omega^+$ and $\omega^-$, because they are in different eigenspaces which commute with the applied operators. The same calculation as above implies $(I_\eta^r + \eta)\pi^+_\eta \omega^-= 0$, which must be invariant under the Lie group generated by $\mathcal{I}_\eta$.  However, there is no subgroup satisfying this condition for all $\eta$, which implies that $\omega$ must be holomorphic.
\end{proof}

Using this lemma, we conclude that holomorphicity on the correspondence space descends to holomorphicity on the twistor space, and the converse.

 For the discussion that follows, it will also be useful to consider the right action of quaternion multiplication, which  leads to the consideration of ``ambitwistor'' space.  Ambitwistor space was introduced to extend the ideas of twistor theory to study general solutions of the Yang-Mills field equations \cite{WITTEN1978394}\cite{ISENBERG1978462}.  The analysis that follows will yield some interesting parallels to this theory.  Here we will introduce some of the details of ambitwistor space in this context.
 
 To begin, we first define a $\beta$-plane similarly to an $\alpha$-plane, by a set of points $q$ satisfying 
 \begin{equation*}
 (q-c) \pi_\eta^+= 0, 
\end{equation*}
for some $c\in\mbb{H}_{\mbb{C}}$, $\eta\in\mbb{I}_u$.  Ambitwistor space refers to the set of null lines in four complex dimensions, denoted $P\mbb{A}$, which can be thought of as the intersection of an $\alpha$-plane and a $\beta$-plane, and are invariant under the translation $q \rightarrow q+\pi_{\eta_1}^-\delta\pi_{\eta_2}^-$.  Defining an $\alpha$-plane by a pair $(c_1,\eta_1)$, and a $\beta$-plane by a pair $(c_2,\eta_2)$, these planes intersect if and only if the condition  
 \begin{equation*}
 \pi_{\eta_1}^+(c_1-c_2)\pi_{\eta_2}^+= 0 
\end{equation*}
holds.  This condition is analogous to the usual condition $Z\cdot W=0$ for an $\alpha$-plane Z and $\beta$-plane W regarded as points in $\text{CP}^3$, with the dot product taken in four dimensions.  Motivated by this condition, we define
\begin{mydef}
Given $\eta_1,\eta_2\in\mbb{I}_u$, and a point $q_0\in \mbb{H}_\mbb{C}$, define a null hyperplane as the set of points satisfying the condition
\begin{equation*}
\pi_{\eta_1}^+(q-q_0)\pi_{\eta_2}^+=0.
\end{equation*}
This defines a three complex dimensional space, which can also be thought of as the linear union of an $\alpha$-plane and a $\beta$-plane that intersect. 
\end{mydef}

Using this definition, one can parameterize the space of null lines through a construction similar to the Klein correspondence, which is reviewed in appendix \ref{sec:klein}. In this case, fix a null hyperplane $H$, which we take to intersect the origin.  Take a null line $L\in P\mbb{A}$ which is the intersection of an $\alpha$-plane $A$ and $\beta$-plane $B$.   For general $L$, both $A$ and $B$ will intersect $H$ on null lines $L_A,L_B\subset H$.  These null lines will then intersect in a point, which is the point of intersection between $L$ and $H$.  Explicitly, if $H$ is defined by  $(\eta_1,\eta_2)$, and $L$ is defined by $(\mu_1,\mu_2)$ and intersects a point $q$, then define $\delta=(\eta_1-\mu_1)^{-1}q(\eta_2-\mu_2)^{-1}$.  Using the property $\pi_\eta^+ \ci = \pi_\eta^+ \eta$, one finds that the point $q-\pi_{\mu_1}^-\delta \pi_{\mu_2}^-$ lies on the hyperplane $H$, and therefore is the intersection point between $H$ and $L$.  If $\eta_1=\mu_1$ or $\eta_2=\mu_2$, however, this construction is not well defined.  This therefore defines a coordinate patch on the space $P\mbb{A}$, and choosing a different hyperplane defines a different coordinate patch, which is a double cover of $P\mbb{A}$.

\section{Holomorphic maps of ambitwistor spaces}

\subsection{Frustrated conformal transformations}

In AHS proposition 3.1 \cite{10.2307/79638}, a general result was derived concerning the existence of solutions to Dirac type equations, which was used to establish a correspondence between self-dual Yang-Mills fields on a self-dual manifold and holomorphic bundles on a complex manifold, extending the ideas of Penrose and Ward \cite{Atiyah1977}.  The operators in this theorem were defined using a symbol operator $\sigma$, which projects onto a subspace of the image of a connection $\nabla$.  In that context, the symbol was defined on Clifford algebra valued differential forms, and used to define the Dirac and twistor operators.

In the present context, we will define a symbol operator $\bar{\sigma}$ as projecting onto the complement of the subspace $\lambda \di q$, where $\lambda$ is a complex scalar.   To motivate the use of this symbol operator, rewrite equation \eqref{eq:tconfeqcp} in the form
\begin{equation}\label{eq:snvk}
\bar{\sigma}\nabla f = 0,
\end{equation}
where $f=\chi^{-1} \kappa \psi^{-1}$ and  $\nabla$ is defined by $\nabla f=\di f+  \di q \xi f + f \xi \di q$.   To generalize this equation, we consider a connection $\nabla$ taking values in the Lie algebra of the conformal orthogonal group, acting as $\nabla f =\di f+ A f + f B$.  To characterize solutions of the equation $\bar{\sigma}\nabla f =0$, we follow AHS proposition 3.1, and define $S_1$ to be the kernel of $\bar{\sigma}$, which are elements of the form $\lambda\di q$.  
The bundle $S_2\subset \mbb{H}_\mbb{C}\otimes \Lambda^2$ denotes the image of $S_1\otimes \Lambda^1$ under exterior multiplication.  There it was shown that this equation is integrable if and only if $D_1 \Gamma(S_1)\in \Gamma(S_2)$ and $\Omega\Gamma(E)\in\Gamma(S_2)$, where $D_1$ is the extension of $\nabla$ to one forms, $\Omega$ is the curvature of $\nabla$, and $E$ is the fibre bundle.  We first characterize $S_2$ by
\begin{mylemma}
The bundle $S_2$ is given by
\begin{equation*}
S_2=\{\omega\in\Gamma(\mbb{H}_\mbb{C}\otimes \Lambda^2):\omega=\frac{1}{2}(\bar{\gamma}\di \bar{q}\wedge\di q - \di q \wedge\di \bar{q}\bar{\gamma}),\gamma\in \Gamma(\mbb{H}_\mbb{C})\},
\end{equation*}
\end{mylemma}
\begin{proof}
An arbitrary real one form $r\in\Gamma(\Lambda^1)$ can be written 
$r=\di q \gamma+\bar{\gamma}\di \bar{q}$ for $\gamma\in\mbb{H}_\mbb{C}$, and an arbitrary element of $S_1$ can be written $\lambda\di q$. We have
\begin{equation*}
\lambda(\di q \gamma+\bar{\gamma}\di \bar{q})\wedge \di q=\frac{1}{2}\lambda(\bar{\gamma}\di \bar{q}\wedge\di q - \di q \wedge\di \bar{q}\bar{\gamma}).
\end{equation*} 
 Absorbing $\lambda$ into $\gamma$ yields the form stated.
 \end{proof}

 The theorem we wish to prove is
\begin{mythm}\label{frustconf}
The equation $\bar{\sigma}\nabla f=0$ is integrable if and only if there is a gauge where $\nabla$ takes the form
\begin{equation}\label{eq:nablaeqn}
\nabla f=\di f+\di q \xi f + f \xi \di q,
\end{equation}
and $\xi$ satisfies the equations
\begin{equation}\label{eq:xicond}
(\partial_{\bar{q}} -\bar{\xi})\di\bar{q}\wedge\di q\xi
=0, \  \text{Im}(\partial_{\bar{q}}\xi)=0.
\end{equation}
This makes the left connection $\nabla_\ell=\di +\di q \xi$ self-dual and the right connection $\nabla_r=\di + \xi \di q$ anti-self-dual.  A solution to these equations will be called a frustrated conformal transformation.
\end{mythm}
\begin{proof}
By AHS proposition 3.1, we must prove that $D_1 \Gamma(S_1)\in \Gamma(S_2)$ and $\Omega\Gamma(E)\in\Gamma(S_2)$ iff the conditions of the proposition hold. For the first condition, we write an arbitrary connection in the form $\nabla f =\di f+ A f + f B$.  Note that the operator $D_1$ acts as $D_1\omega = \di \omega + A\wedge \omega - \omega \wedge B$ on a section $\omega\in\Gamma(E\otimes T^*M)$.  Write $A=\sum_\mu \e_\mu \di q A_\mu$ and $B=\sum_\mu B_\mu \di q \e_\mu$, take a form $\lambda \di q \in \Gamma(S_1)$ and apply the operator $D_1$, to find 
\begin{align*}
D_1 (\lambda \di q) = \di \lambda  \wedge \di q +\lambda \sum_\mu (\e_\mu \di q \wedge A_\mu \di q-&\di q \wedge B_\mu \di q \e_\mu)= \nonumber\\
 -\frac{\lambda }{2}\sum_\mu  \Big(\e_\mu \bar{A}_\mu\di \bar{q}\wedge\di q  - \bar{B}_\mu \di \bar{q}\wedge\di q  \e_\mu + &\e_\mu \di q \wedge \di \bar{q} \bar{A}_\mu  - \di q \wedge \di \bar{q} \bar{B}_\mu \e_\mu\Big)\nonumber\\
&+\frac{1}{2}(\partial_{\bar{q}}\lambda\,\di \bar{q}\wedge \di q- \di q \wedge \di \bar{q}\, \partial_{\bar{q}}\lambda).
\end{align*}
 The last expression is already of the desired form.  The term $\e_\mu \di q \wedge \di \bar{q} \bar{A}_\mu$ is of the desired form only if $A_\mu=\frac{1}{4}(2\xi_1\delta_{\mu 0}+\xi_2\e_\mu) $ for some $\xi_1,\xi_2\in\mbb{H}_\mbb{C}$, which implies that $A=\frac{1}{2}(\di q \xi_1-\bar{\xi}_2\di \bar{q})$, and similarly $B=\frac{1}{2}(\xi_1'\di q - \di \bar{q}\bar{\xi}_2')$.  For the other two terms, one must demand that $\xi_{1,2}=\xi_{1,2}'$, then this condition is satisfied.

To prove the second part, we calculate the curvature of the connection $\nabla$.  Acting on a function $f\in\Gamma(E)$, we find 
\begin{equation*}
\Omega \, f= \Omega_\ell  f + f \Omega_r, 
\end{equation*}
where $\Omega_\ell=\di A + A\wedge A$ and $\Omega_r = \di B - B\wedge B$.  Calculating $\Omega_\ell$ yields
\begin{align*}
\Omega_\ell=\frac{1}{4}\Big(-\bar{\xi}_2\overleftarrow{\partial_{q}}\di q & \wedge\di \bar{q}-\di q \wedge \di \bar{q}\big(\partial_{\bar{q}}\xi_1+\frac{1}{2}\lvert \xi_1+\xi_2\rvert^2\big) \nonumber\\
&+\bar{\xi}_2\di\bar{q}\wedge\di q\overleftarrow{\partial_q}+\partial_{\bar{q}}\di \bar{q}\wedge\di q\xi_1
-\frac{1}{2}(\bar{\xi}_1+\bar{\xi}_2)\di \bar{q}\wedge\di q (\xi_1+\xi_2)\Big).
\end{align*}
Here we have used the identity $\di q\wedge \di f=\frac{1}{2}\Big(\di q \wedge \di \bar{q}\,\partial_{\bar{q}}-\partial_{\bar{q}}\,\di \bar{q}\wedge \di q\Big)f$,
which can be derived using the techniques of section \ref{sec:qanalysis}. From the first term, it is clear that $\Omega\Gamma(E)\in\Gamma(S_2)$ only if $\text{Im}(\partial_{\bar{q}}\xi_2)=0$.   A similar calculation for $\Omega_r$ also requires that $\text{Im}(\partial_{q}\bar{\xi}_2)=0$. This is the integrability condition for a real connection $\nabla_R = \di + \frac{1}{2}(\di q \xi_2 +\bar{\xi}_2\di \bar{q})$, with curvature $\Omega_R = -\frac{1}{2}\text{Re}(\di \bar{q}\wedge \di q \partial_q \bar{\xi}_2+\di q \wedge \di \bar{q} \partial_{\bar{q}}\xi_2)$, which vanishes if $\text{Im}(\partial_{\bar{q}}\xi_2)=\text{Im}(\partial_{q}\bar{\xi}_2)=0$.   This means there is a real gauge transformation that absorbs $\xi_2$.  Performing this gauge transformation, the left gauge field $A$ takes the form $A=\di q \xi$, and the right gauge field takes the form $B=\xi \di q$. 

For the gauge transformed field $A=\di q \xi$, the curvature then becomes
\begin{align}\label{eq:nabcurv}
\Omega_\ell=\frac{1}{2}\Big(-\di q \wedge \di \bar{q}(\partial_{\bar{q}}+\bar{\xi})\xi+(\partial_{\bar{q}}-\bar{\xi})\di \bar{q}\wedge\di q\xi\Big).
\end{align}
If $f$ has an imaginary part, then in order for $\Omega\Gamma(E)\in\Gamma(S_2)$, we must demand that $\Omega_\ell$ is self-dual, because $\Omega_\ell$ is multiplied on the left\footnote{Here we assume that points where $f$ is real are isolated and therefore excluded as pathological cases, although this could merit further investigation.}. By a similar calculation, the right curvature is given by 
\begin{equation*}
\Omega_r = \frac{1}{2}\Big(\xi\left(\overleftarrow{\partial}_{\bar{q}}+\bar{\xi}\right)\di \bar{q}\wedge \di q - \xi \di q \wedge \di \bar{q} \left(\overleftarrow{\partial}_{\bar{q}}-\bar{\xi}\right)\Big),
\end{equation*}  
and this must be anti-self-dual.  Note that equations \eqref{eq:xicond} are sufficient to satisfy both of these conditions.  To see this, write a particular imaginary component of the first equation in the form
\begin{equation*}
\text{Re}(\e_j(\partial_{\bar{q}}-\bar{\xi})\e_i \xi)=0.
\end{equation*}  
Using the cyclic property, one can rearrange this to satisfy both imaginary parts.  Taking the real part of the first equation implies $\text{Im}(\partial_q \bar{\xi})=0$.   The second equation in \eqref{eq:xicond} then makes $\Omega_r$ anti-self-dual.  
If these conditions are satisfied, then $\Omega\Gamma(E)\in\Gamma(S_2)$ and the conditions of the theorem hold.  
\end{proof}

We use the term frustrated conformal transformation because if the connection $\nabla$ was integrable, the resulting function $\kappa=\chi f \psi$ would be conformal.  We note that the resulting equations have some similarity to the equations introduced in the study of classical Yang-Mills solutions with ambitwistor methods \cite{WITTEN1978394}.  In that case, the space was enlarged to eight dimensions, where the left and right connections acted on the two copies of complexified Minkowski space.  However, on the diagonal subspace, these equations take on a similar form, with the added ``torsion condition'' \eqref{eq:nablaeqn}.  It is not clear what the exact relationship between these cases are, but ambitwistor methods are indeed useful to interpret the solutions of these equations, which is explored further in the following section.

\subsection{Maps on ambitwistors}

The previous construction has a natural interpretation in the context of ambitwistors, which appeals to a technique similar to the Penrose-Ward transform.  Here we will show that the conditions derived previously for a frustrated conformal transformation imply that the induced map on ambitwistors preserves null lines.  We then sketch an argument for how to make this construction valid globally, although more work is necessary to make this precise.

To proceed, note that the connection $\nabla$ introduced in the previous section is integrable on null lines.  This condition is trivial, because the contraction of any two-form  with a vector of the form $\pi_{\eta_1}^-\delta \pi_{\eta_2}^-$ vanishes, which can be readily checked.  Therefore, the functions $\chi$ and $\psi$ are well defined on the null line, because the curvature vanishes.  So on this line one can define a function 
\begin{equation*}
\kappa = \chi f \psi.
\end{equation*} 
Define the ambitwistor correspondence space $\mathcal{F}_A =\mbb{I}_{u,\ell}\otimes \mbb{H}_\mbb{C} \otimes \mbb{I}_{u,r}$, where $\mbb{I}_{u,\ell}$ denotes the left  complex structure parameterized by $\eta_1$, and similarly for $\mbb{I}_{u,r}$.  $P\mbb{A}$ is then obtained by identifying points satisfying $q_1 = q_2+ \pi_{\eta_1}^- \delta \pi_{\eta_2}^-$.  The function $\chi$ then induces a map $m_{\eta_1}:\mbb{I}_{u,\ell}\rightarrow \mbb{I}_{u,\ell}$ defined by $m_{\eta_1}=\chi^{-1}\eta_1 \chi$, and similarly $m_{\eta_2}=\psi \eta_2 \psi^{-1}$.  From equation \eqref{eq:tlineeq1},  the demand that $m_{\eta_1}$ is constant on $\alpha$-planes is equivalent to the condition
\begin{equation*}
\pi_\eta^+[\eta, A(\pi_\eta^-\delta)]=0, \ \ \forall \delta\in\mbb{H}.
\end{equation*}
 From lemma \ref{hollemma}, if the form $A$ is holomorphic, then this is equivalent to the condition
\begin{equation*}
\mathcal{I}_\eta\pi_\eta^+[\eta, A]=0.
\end{equation*}
Based on the arguments of section \ref{sec:cco}, this implies
\begin{equation}\label{eq:Aeqn}
A=\frac{1}{2}(\di q \xi_1-\bar{\xi}_2\di \bar{q}).
\end{equation}
A similar analysis shows $B=\frac{1}{2}( \xi_1\di q-\di \bar{q}\bar{\xi}_2)$.  Together, these conditions imply that $m_{\eta_1}$ and $m_{\eta_2}$ are constant on null lines.  

The function $\kappa$ preserves null lines only if it preserves $\alpha$ and $\beta$ planes when restricted to a null line.  The condition that it preserves $\alpha$ planes gives
\begin{equation*}
 \pi_{\eta_1}^+  \chi^{-1} \di \kappa(\pi_{\eta_1}^-\delta)=0,
\end{equation*}
where we have allowed dependence of $\kappa$ on $\eta$. 
This can be rewritten
\begin{equation}\label{eq:kappaeqthol}
 \pi_{\eta_1}^+  \nabla f(\pi_{\eta_1}^-\delta)=0, \ \ \forall \delta \in \mbb{H},
\end{equation}
In this case, the analysis leading to equation \eqref{eq:kappaeq} implies
\begin{equation*}
 \nabla f = \di q \nu,
\end{equation*}
for some $\nu\in\mbb{H}_\mbb{C}$. 
Similarly, if we demand that the map preserves $\beta$ planes, then this implies $ \nabla f = \nu \di q$.  Combining these conditions, we find 
\begin{equation*}
\bar{\sigma}\nabla f = 0,
\end{equation*}
which is the equation for a frustrated conformal transformation introduced in the previous section. We therefore find that the conditions considered previously are sufficient to imply that the induced maps preserve null lines, and locally give rise to mappings from $P\mbb{A}$ to itself.  Because these maps are holomorphic, they are preserved locally under composition.  

For the global considerations,  one can view the ambitwistor holomorphic map as a section of a fibre bundle.  This bundle is a product bundle, where the product contains the  left and right copies of $Sp(1)$, and an affine bundle for the function $\kappa$.   Now introduce a double cover of ambitwistor space as defined in section \ref{sec:rtwist}.  Introduce two null hyperplanes $S$ and $T$.  The coordinates of the two patches will be defined by the intersection of $L\in P\mbb{A}$ with $S$ and $T$ at points $q_S$ and $q_T$.  The patching relations for $\chi$ and $\psi$ are then the path ordered exponentials
\begin{equation*}
F_\chi = \tilde{\chi}^{-1} \chi= \mathcal{P} \exp\left(\int_{q_T}^{q_S} A\right) , \ \ F_\psi = \psi \tilde{\psi}^{-1}= \mathcal{P} \exp\left(\int_{q_S}^{q_T} B\right).
\end{equation*}
On $L$ one can again define $\kappa = \chi f \psi$, which satisfies $\di \kappa = \chi  \lambda\di q \psi$ for some $\lambda\in \mbb{C}$.   Propagating the equation $\bar{\sigma}\nabla f=0$ then yields a patching relation for $\kappa$ 
\begin{equation*}
\tilde{\kappa}=F_\chi \kappa F_\psi+\tilde{\chi}(\tilde{f}-f)\tilde{\psi},
\end{equation*}
which takes the form of an affine transformation.  This is why the relevant bundle is an affine bundle.  It will require more investigation, however, to further understand the details of this construction.

\subsection{An example}

In this section we will work out the details for the simplest nontrivial example of an ambitwistor holomorphic map, based on the ``basic instanton'' of Belavin et. al. \cite{BELAVIN197585} and the quaternionic formula of Atiyah \cite{atiyah1979geometry}.  In this context, the basic instanton is a self-dual connection defined by the gauge field
\begin{equation*}
A=-\frac{\di q \,\bar{q}}{1+\lvert q\rvert^2}.
\end{equation*}
Note that usually the imaginary part of this potential is taken to preserve the unitary structure, but based on the discussion above we include the real part.  This connection has the correct form to define the left connection $\nabla_{\ell}$ for a frustrated conformal transformation, with $\xi=-\bar{q}/(1+\lvert q\rvert^2)$.  Based on theorem \ref{frustconf}, the right connection will be defined by $B=\xi \di q$.  The curvature of the left and right connections are given by 
\begin{equation*}
\Omega_\ell=\frac{\di q \wedge \di \bar{q}}{(1+\lvert q\rvert^2)^2}, \ \ \ \Omega_r=-\frac{\di \bar{q} \wedge \di q}{(1+\lvert q\rvert^2)^2},
\end{equation*}
verifying that the left connection is self-dual and the right connection is anti-self-dual.

In order to solve $\bar{\sigma}\nabla f=0$, we first restrict the equation to a null line, for which we can introduce $\kappa=\chi f \psi$, where $\nabla_\ell \chi^{-1} = 0$ and $\nabla_r \psi^{-1}=0$.  $\kappa$ satisfies the equation $\di \kappa = \chi \lambda \di q \psi$ on this line, which can be integrated if the function $\lambda$ is known.   To find $\lambda$, we expand and simplify $\di^2 f=0$ to find
\begin{equation*}
\frac{1}{2}\di \lambda = -\frac{\di q \bar{f}+f\di \bar{q}}{(1+\lvert q\rvert^2)^2},
\end{equation*}
which depends on $f$.  Restricting to a null line through the origin, and choosing the initial condition $f(0)=0$, these equations take on a particularly simple form.  The connections $\nabla_{\ell,r}$ are trivial, and $\di \lambda=0$.  Therefore we find that a particular solution on null lines through the origin is $f = \kappa = q$.  In fact, it is straightforward to verify that $f=q$ is a solution for all $q$.  The equation is conformally invariant, so any conformal transformation is also a solution.  

We see that the function $f$ is still conformal, and therefore not particularly interesting at first sight.  However, the induced map on ambitwistors is non-trivial, and therefore does not reduce to a conformal transformation.  In order to find this map, one must solve for $\chi$ and $\psi$ in a certain coordinate patch.  This has been accomplished for various self-dual connections, including the BPST instanton, as in \cite{8d0ab9a72ff7460e86eec538b76ce602}.  However, there are some tricky issues concerning the gauge, and it would be desirable to find a simpler expression for these functions in terms of quaternions.  This will be pursued further in a later publication.

\section{Conclusion}

We have defined an equation for a frustrated conformal transformation, which contains a self-dual connection and anti-self-dual connection in its definition.  These functions give rise to maps on ambitwistor space, and because these maps are closed under composition, it is possible that this could provide a geometric way to generate solutions to the self-dual Yang-Mills equations.  Furthermore, the study of these functions could give rise to a new understanding of ambitwistor space and its base space.  It will be interesting to analyze the solution space of these equations, to understand in more depth the types of transformations that arise.  Furthermore, these maps could possibly lead to more natural coordinate systems for four dimensional quaternionic manifolds.  Although more work needs to be done, these initial results appear promising.

\appendix

\section{Embedding of the Riemann sphere}
\label{sec:emriemann}
Here we will show some intermediate steps leading to equation \eqref{eq:tline} for holomorphic embeddings $\mbb{I}_u\rightarrow \mbb{H}\otimes\mbb{I}_u$.
First, define a map $v(z_1,z_2)=z_1+z_2j$, noting that equation \eqref{eq:stereographic2} can be written $\conj_{v(z_1,z_2)}(\bi)$, where $\conj_\chi(\eta)=\chi^{-1}\eta\chi$.  Under this map, a M{\"o}bius transformation is given by $\chi\rightarrow \chi \alpha + \bi\chi\beta$, where $\alpha$ and $\beta$ are defined according to equation \eqref{eq:etamob}.  Extending this map to embeddings of the form \eqref{eq:cp1tocp3}, we can write
\begin{equation}\label{eq:chimap}
\chi\rightarrow(\chi\tilde{\alpha}+\bi\chi\tilde{\beta},\chi\alpha+\bi\chi\beta).
\end{equation}
In this context, the image is only defined up to simultaneous left multiplication by $u\in\mbb{C}^*$, denoted $\mbb{H}^2/\mbb{C}^*$. Define a map $\tilde{\phi}:\mbb{H}^2/\mbb{C}^*\rightarrow\mbb{H}\otimes\mbb{I}_u$ given by 
\begin{equation*}
\tilde{\phi}(q_1,q_2)=\Big(q_2^{-1}q_1,\conj_{q_2}(\bi)\Big),
\end{equation*}
which is cleary related to equation \eqref{eq:phitriv}.  
Applying $\tilde{\phi}$ to the image of equation \eqref{eq:chimap}, we obtain a map $\tau_q$ defined by
\begin{equation*}
\tau_q\big(\conj_\chi(i)\big)\rightarrow\Big((\chi\alpha+\bi\chi\beta)^{-1}(\chi\tilde{\alpha}+\bi\chi\tilde{\beta}),\conj_{\chi\alpha+\bi\chi\beta}(\bi)\Big).
\end{equation*}
Defining $\eta=\conj_\chi(i)$ and rewriting this function in terms of $\eta$, we obtain the desired result \begin{equation*}
\tau_q(\eta)=\Big(h+\mathcal{M}_{\alpha,\beta}(\eta)\rho,\mathcal{M}_{\alpha,\beta}(\eta)\Big),
\end{equation*}
where $h=\alpha^{-1}(\beta\alpha^{-1}+\alpha\beta^{-1})^{-1}(\alpha\beta^{-1}+\tilde{\beta}\tilde{\alpha}^{-1})\tilde{\alpha}$, $\rho=\beta^{-1}(\beta\alpha^{-1}+\alpha\beta^{-1})^{-1}(\tilde{\beta}\tilde{\alpha}^{-1}-\beta\alpha^{-1})\tilde{\alpha}$, and $\alpha$ and $\beta$ are defined similarly as before.

\section{Holomorphic vector bundles and Klein correspondence}\label{sec:klein}

Here we review some aspects of holomorphic vector bundles and the Klein correspondence.  More details can be found in the books by Ward \cite{ward1991twistor} and Dunajski \cite{dunajski2010solitons}. A holomorphic vector bundle is defined by a set of local trivializations $(U_\alpha,\chi_\alpha)$ such that the transition maps $F_{\alpha\beta}$ are holomorphic on intersections.  A bundle is said to be trivial if the transition maps satisfy $F_{\alpha\beta}=H_{\beta}H_{\alpha}^{-1}$ for holomorphic functions $H_\alpha,H_\beta$ defined on $U_\alpha,U_\beta$.  On the Riemann sphere, one can introduce a double cover $U=\{(\pi_0,\pi_1)\in\mbb{C}\mbb{P}_1:\pi_1\neq 0\}$ and $\tilde{U}=\{(\pi_0,\pi_1)\in\mbb{C}\mbb{P}_1:\pi_0\neq 0\}$, and the transition function of a trivial bundle then satisfies $F=\tilde{H}H^{-1}$.  A global section of a trivial bundle is given by  a function $s_\alpha(z) = H_\alpha(z) \xi$ for a constant vector $\xi$, and these are the only globally holomorphic sections.

To define the Klein correspondence, we use  quaternionic notation, which directly parallels the construction with complex coordinates.  Consider two fixed points $P,Q\in\mathcal{P}\mathcal{T}$ corresponding to $\alpha$-planes $\tilde{P},\tilde{Q}$ defined by $\pi_{\bi}^{\pm} q=0$.  Consider a third point $Z\in\mathcal{P}\mathcal{T}$ corresponding to an $\alpha$-plane $\tilde{Z}$ defined by $\pi_\eta^+ (q-\delta)=0$.  The intersection of $\tilde{Z}$ with $\tilde{P},\tilde{Q}$ occurs at the point $q_Z^{\pm}=(\eta\mp \bi)^{-1}\pi_\eta^+ \delta$, as can be readily checked.  Clearly this is not defined if $\eta=\pm \bi$, and so this defines two coordinate charts for $\mathcal{P}\mathcal{T}$.  If $\eta\neq \pm \bi$, then both intersections $q_Z^{\pm}$ exist and lie on the $\alpha$-plane $\tilde{Z}$, and therefore can be connected by a path $\Gamma_Z$ within $\tilde{Z}$, allowing for the construction of a patching function between the two coordinate patches as in the Penrose-Ward correspondence.

\bibliographystyle{unsrt}
\bibliography{References}

\end{document}